\documentclass{article}
\usepackage{fullpage}

\usepackage[utf8]{inputenc}
\usepackage{epsfig}
\usepackage{epstopdf}
\usepackage{graphicx}
\DeclareGraphicsExtensions{.pdf,.png,.jpg,.eps}

\usepackage{latexsym}
\usepackage{amsfonts}
\usepackage{amssymb}
\usepackage{amsmath}
\usepackage{amsthm}
\usepackage{bm}
\usepackage{algorithm} 
\usepackage{algpseudocode}
\usepackage{color}
\usepackage{multirow}
\usepackage{url}
\usepackage{pifont}

\newcommand{\MS}[1]{\ensuremath{\Delta}}

\newcommand{\RMinQ}{{\textsf{RMinQ}}}
\newcommand{\RLMinQ}{{\textsf{RLMinQ}}}
\newcommand{\RRMinQ}{{\textsf{RRMinQ}}}
\newcommand{\RkMinQ}{{\textsf{RkMinQ}}}

\newcommand{\RMaxQ}{{\textsf{RMaxQ}}}
\newcommand{\RLMaxQ}{{\textsf{RLMaxQ}}}
\newcommand{\RRMaxQ}{{\textsf{RRMaxQ}}}
\newcommand{\RkMaxQ}{{\textsf{RkMaxQ}}}
\newcommand{\NLN}{{\textsf{NLN}}}
\newcommand{\ANLN}{{\textsf{ANLN}}}

\newcommand{\PSV}{{\textsf{PSV}}}
\newcommand{\NSV}{{\textsf{NSV}}}
\newcommand{\PLV}{{\textsf{PLV}}}
\newcommand{\NLV}{{\textsf{NLV}}}

\newcommand{\NRS}{{\textsf{NRS}}}

\newcommand{\rank}{{\textsf{rank}}}
\newcommand{\select}{{\textsf{select}}}
\newcommand{\findopen}{{\textsf{findopen}}}
\newcommand{\findclose}{{\textsf{findclose}}}

\newcommand{\leftmost}{\mathit{leftmost}}

\newcommand{\MinA}{{\textsf{Min(A)}}}
\newcommand{\MaxA}{{\textsf{Max(A)}}}
\newcommand{\CMinA}{{\textsf{CMin(A)}}}
\newcommand{\CMaxA}{{\textsf{CMax(A)}}}

\newcommand{\MinAp}{{\textsf{Min(A')}}}

\newcommand{\CMinAp}{{\textsf{CMin(A')}}}

\newtheorem{theorem}{Theorem}
\newtheorem{lemma}{Lemma}[section]

\begin{document}

\title{Simultaneous encodings for range and next/previous larger/smaller value queries\footnote{Preliminary version of these results have appeared in the proceedings of the 21st International Computing and Combinatorics Conference (COCOON-2015)~\cite{DBLP:conf/cocoon/JoS15}}}
\author{ Seungbum Jo and Srinivasa~Rao~Satti \\
Seoul National University, South Korea \\
{\tt sbcho,ssrao@tcs.snu.ac.kr}
}
%
\date{}
\maketitle
\begin{abstract}
Given an array of $n$ elements from a total order, we propose encodings 
that support various range queries (range minimum, range maximum 
and their variants), and previous and next smaller/larger value queries.
When query time is not of concern, we obtain a $4.088n + o(n)$-bit 
encoding that supports all these queries. 
For the case when we need to support all these queries in constant time,
we give an encoding that takes $4.585n + o(n)$ bits, where $n$ is the length 
of input array. 
%
This improves the $5.08n+o(n)$-bit encoding obtained by encoding the colored $2d$-Min and Max heaps proposed by Fischer~[TCS, 2011].
We first extend the original DFUDS~[Algorithmica, 2005] 
encoding of the colored 2d-Min (Max) heap 
that supports the queries in constant time.
Then, we combine the extended DFUDS of $2d$-Min heap and $2d$-Max heap using the 
Min-Max encoding of Gawrychowski and Nicholson~[ICALP, 2015] 
with some modifications.
We also obtain encodings that take lesser space and support a subset of these queries.
\end{abstract}

\section{Introduction}
\label{sec:intro}
Given an array $A[1 \dots n]$ of $n$ elements from a total order.
For $1 \leq i \leq j \leq n$, suppose that there are $m$ ($l$) positions 
$i \leq p_1 \leq \dots \leq p_m \leq j$ ($i \leq q_1 \leq \dots \leq q_l \leq j$ ) in $A$ 
which are the positions of minimum (maximum) values between $A[i]$ and $A[j]$.
Then we can define various range minimum (maximum) queries as follows.
\begin{itemize}
\item{Range Minimum Query} ($\RMinQ{}_A(i, j)$) : Return an arbitrary position among $p_1, \dots, p_m$.
\item{Range Leftmost Minimum Query} ($\RLMinQ{}_A(i, j)$) : Return $p_1$.
\item{Range Rightmost Minimum Query} ($\RRMinQ{}_A(i, j)$) : Return $p_j$.
\item{Range $k$-th Minimum Query} ($\RkMinQ{}_A(i, j)$) : Return $p_k$ (for $1\le k \le m$). 
\item{Range Maximum Query} ($\RMaxQ{}_A(i, j)$) : Return an arbitrary position among $q_1, \dots, q_l$.
\item{Range Leftmost Maximum Query} ($\RLMaxQ{}_A(i, j)$) : Return $q_1$.
\item{Range Rightmost Maximum Query} ($\RRMaxQ{}_A(i, j)$) : Return $q_l$.
\item{Range $k$-th Maximum Query} ($\RkMaxQ{}_A(i, j)$) : Return $q_k$ (for $1 \le k \le l$). 
\end{itemize}

Also for $1\le i \le n$, we consider following additional queries on $A$. 
\begin{itemize}
\item{Previous Smaller Value} ($\PSV{}_A(i)$) : $\max{(j : j < i, A[j] < A[i])}$.
\item{Next Smaller Value} ($\NSV{}_A(i)$) : $\min{(j : j > i, A[j] < A[i])}$.
\item{Previous Larger Value} ($\PLV{}_A(i)$) : $\max{(j : j < i, A[j] > A[i])}$.
\item{Next Larger Value} ($\NLV{}_A(i)$) : $\min{(j : j > i, A[j] > A[i])}$.
\end{itemize}
For define above four queries formally, we assume that $A[0] = A[n+1] = -\infty$ for $\PSV{}_A(i)$ and $\NSV{}_A(i)$.
Similarly we assume that $A[0] = A[n+1] = \infty$ for $\PLV{}_A(i)$ and $\NLV{}_A(i)$.


Our aim is to obtain space-efficient encodings that support these queries efficiently.
An encoding should support the queries without accessing the input array (at query time).
The minimum size of an encoding is also referred to as the  
 \textit{effective entropy} of the input data (with respect to the queries)~\cite{gikrs-isaac11}.
We assume the standard word-RAM model~\cite{miltersen-survey} with word size $\Theta(\lg{n})$.

\paragraph{Previous Work}  
The range minimum/maximum problem has been well-studied in the literature.
It is well-known~\cite{BenderF00} that finding $\RMinQ{}_A$ can be transformed to the problem of 
finding the LCA (Lowest Common Ancestor) between (the nodes corresponding to) the two query positions in the Cartesian tree constructed on $A$. 
Furthermore, since different topological structures of the Cartesian tree on $A$ give rise to 
different set of answers for $\RMinQ{}_A$ on $A$, one can obtain an 
information-theoretic lower bound of 
$2n-\Theta(\lg{n})$\footnote{We use $\lg{n}$ to denote $\log_2{n}$} bits
on the encoding of $A$ that answers $\RMinQ{}$ queries.
Sadakane~\cite{Sadakane2007a} proposed the $4n+o(n)$-bit encoding with constant query time for $\RMinQ{}_A$ problem
using the balanced parentheses (BP)~\cite{mr-sjc01} of the Cartesian tree of $A$ with some additional nodes. 
Fischer and Heun~\cite{fh-sjc11} introduced the $2d$-Min heap, which is a
variant of the Cartesian tree, and showed how to encode it using the 
Depth first unary degree sequence (DFUDS)~\cite{BDMRRR05} representation
in $2n+o(n)$ bits which supports $\RMinQ{}_A$ queries in constant time.
Davoodi et al. show that same $2n+o(n)$-bit encoding with constant query time
can be obtained by encoding the Cartesian trees~\cite{DavoodiRS12}. 
For $\RkMinQ{}_A$, Fischer and Heun~\cite{FischerH10} defined the
\textit{approximate range median of minima query} problem which returns 
a position $\RkMinQ{}_A$ for some $\frac{1}{16}m \le k \le \frac{15}{16}m$,
and proposed an encoding that uses $2.54n+o(n)$ bits and supports the 
{\it approximate $\RMinQ{}_A$} queries in constant time, using a \textit{Super Cartesian tree}.

For $\PSV{}_A$ and $\NSV{}_A$, if all elements in $A$ are distinct, 
then $2n+o(n)$ bits are enough to answer the queries in constant time, 
by using the $2d$-Min heap of Fischer and Heun~\cite{fh-sjc11}.
For the general case, Fischer~\cite{Fischer11} proposed the \textit{colored $2d$-Min heap},
and proposed an optimal $2.54n+o(n)$-bit encoding which can answer $\PSV{}_A$ and $\NSV{}_A$ in constant time.
As the extension of the $\PSV{}_A$ and $\NSV{}_A$, one can define the \textit{Nearest Larger Neighbor($\NLN{}(i)$)} on $A$
which returns $\PSV{}_A(i)$ if $i-\PSV{}_A(i) \le \NSV{}_A(i) - i$ and returns $\NSV{}_A(i)$ otherwise.
This problem was first discussed by Berkman et al.~\cite{BerkmanSV93} and they proposed a parallel algorithm to answer $\NLN{}$ queries for all positions on the array (this problem is defined as \textit{All-Nearest Larger Neighbor ($\ANLN{}$) problem}.) 
and Asano and Kirkpatrick~\cite{AsanoK13} proposed time-space tradeoff algorithms for $\ANLN{}$ problem.
Jayapaul et al.~\cite{IWOCA2014} proposed $2n+o(n)$-bit encoding 
which supports an $\NLN{}(i)$ on $A$ in constant time if all elements in $A$ are distinct.


One can support both $\RMinQ{}_A$ and $\RMaxQ{}_A$ in constant time trivially using the encodings for $\RMinQ{}_A$ and $\RMaxQ{}_A$ queries, using a total of $4n+o(n)$ bits.
Gawrychowski and Nicholson reduce this space to $3n+o(n)$ bits while maintaining constant time query time~\cite{Gawry14}.
Their scheme also can support $\PSV{}_A$ and $\PLV{}_A$ in constant time when there are no consecutive equal elements in $A$.

\paragraph{Our results}
In this paper, we first extend the original DFUDS~\cite{BDMRRR05} for colored 2d-Min(Max) heap 
that supports the queries in constant time.
Then, we combine the extended DFUDS of $2d$-Min heap and $2d$-Max heap using 
Gawrychowski and Nicholson's Min-Max encoding~\cite{Gawry14} with some modifications.
As a result, we obtain the following non-trivial encodings that support a wide range of queries. 


\begin{theorem}
\label{theorem:minmax}
An array $A[1 \dots n]$ containing $n$ elements from a total order can be encoded using
\begin{enumerate}
\item[(a)]  at most $3.17n+o(n)$ bits to support $\RMinQ{}_A$, $\RMaxQ{}_A$, $\RRMinQ{}_A$, 
$\RRMaxQ{}_A$,  $\PSV{}_A$, and $\PLV{}_A$ queries;
\item[(b)]  at most $3.322n+o(n)$ bits to support the queries in (a) in constant time;
\item[(c)] at most $4.088n+o(n)$ bits to support $\RMinQ{}_A$, $\RRMinQ{}_A$, 
$\RLMinQ{}_A$, $\RkMinQ{}_A$, $\PSV{}_A$, $\NSV{}_A$, $\RMaxQ{}_A$, 
$\RRMaxQ{}_A$, $\RLMaxQ{}_A$, $\RkMaxQ{}_A$, $\PLV{}_A$ and $\NLV{}_A$ queries; and
\item[(d)]  at most $4.585n+o(n)$ bits to support the queries in (c) in constant time.
\end{enumerate}
If the array contains no two consecutive equal elements, then (a) and (b) take $3n+o(n)$ bits, and
(c) and (d) take $4n+o(n)$ bits.
\end{theorem}

This paper organized as follows. Section~\ref{sec:prelim} introduces various data structures 
that we use later in our encodings. In Section~\ref{sec:exdfuds}, 
we describe the encoding of colored $2d$-Min heap 
by extending the DFUDS of $2d$-Min heap. This encoding uses a distinct approach from the encoding of the colored $2d$-Min heap by Fischer~\cite{Fischer11}.
Finally, in Section~\ref{sec:minmax}, we combine the encoding of this colored $2d$-Min heap and 
Gawrychowski and Nicholson's Min-Max encoding~\cite{Gawry14} with some modifications,
to obtain our main result (Theorem~\ref{theorem:minmax}).

\section{Preliminaries}
\label{sec:prelim}

We first introduce some useful data structures that we use to encode various bit vectors and balanced parenthesis sequences.


\paragraph{\textbf{Bit strings (parenthesis sequences)}}
Given a string $S[1 \dots n]$ over the alphabet $\Sigma = \{(,  )\}$,
$\rank{}_{S}(i, x)$ returns the number of occurrence of the pattern $x \in \Sigma^*$ in $S[1 \dots i]$
and $\select{}_{S}(i, x)$ returns the first position of $i$-th occurrence of $x \in \Sigma^*$ in $S$.
Combining the results from \cite{MRR01} and \cite{rrr-talg07}, one can show the following.
%
%

\begin{lemma}[\cite{MRR01}, \cite{rrr-talg07}]
\label{lemma:rankselect2}
Let $S$ be a string of length $n$ over the 
alphabet $\Sigma = \{'(',  ')'\}$ containing $m$ closing parentheses.
One can encode $S$ using $\lg{{n \choose m}}+o(n)$ bits to support both $\rank{}_{S}(i, x)$ 
and $\select{}_{S}(i, x)$ in constant time, for any pattern $x$ with length $|x| \leq 1/2\lg{n}$. 
Also, we can decode any $\lg{n}$ consecutive bits in $S$ in constant time.
\end{lemma}

\paragraph{\textbf{Balanced parenthesis sequences}} 
Given a string $S[1 \dots n]$ over the alphabet $\Sigma = \{'(',  ')'\}$,
if $S$ is balanced  
and $S[i]$ is an open (close) parenthesis, 
then we can define $\findopen{}_{S}(i)$ ($\findclose{}_{S}(i)$) 
which returns the position of the matching close (open) parenthesis to $S[i]$. 
Now we introduce the lemma from Munro and Raman~\cite{mr-sjc01}.
\begin{lemma}[\cite{mr-sjc01}]
\label{lemma:openclose}
Let $S$ be a balanced parenthesis sequence of length $n$.
If one can access any $\lg{n}$-bit subsequence of $S$ in constant time, 
Then both $\findopen{}_{S}(i)$ and $\findclose{}_{S}(i)$ can be supported in constant time
with $o(n)$-bit additional space.
\end{lemma}
\paragraph{\textbf{Depth first unary degree sequence}}

\textit{Depth first unary degree sequence (DFUDS)} is one of the well-known methods for representing ordinal trees~\cite{BDMRRR05}. 
It consists of a balanced sequence of open and closed parentheses, which can be defined inductively as follows. 
If the tree consists of the single node, its DFUDS is `$()$'. 
Otherwise, if the ordinal tree $T$ has $k$ subtrees $T_1 \dots T_k$, then its DFUDS, $D_T$ is 
the sequence $(^{k+1}~) d_{T_1} \dots d_{T_k}$ (i.e., $k+1$ open parentheses followed by a close parenthesis concatenated with the `partial' DFUDS sequences $d_{T_1} \dots d_{T_k}$) where $d_{T_i}$, for $1 \le i \le k$, is the DFUDS of the subtree $T_i$ 
(i.e., $D_{T_i}$) with the first open parenthesis removed.
From the above construction, it is easy to prove by induction that if $T$ has $n$ nodes, then the size of $D_T$ is $2n$ bits. 
%
The following lemma shows that DFUDS representation can be used to support various navigational operations on the tree efficiently.
\begin{lemma}[\cite{ArroyueloCNS10}, \cite{BDMRRR05}, \cite{JSS07}] 
\label{lemma:operation-dfuds}
Given an ordinal tree $T$ on $n$ nodes with DFUDS sequence $D_T$, one can construct
an auxiliary structure of size $o(n)$ bits to support the following operations in constant time: 
for any two nodes $x$ and $y$ in $T$, 

\noindent\emph{-} $parent_{T}(x)$ : Label of the parent node of node $x$. \\
\noindent\emph{-} $degree_{T}(x)$ : Degree of node $x$. \\
\noindent\emph{-} $depth_{T}(x)$ : Depth of node $x$ (The depth of the root node is 0). \\
\noindent\emph{-} $subtree\_size_{T}(x)$ : Size of the subtree of $T$ which has the $x$ as the root node. \\
\noindent\emph{-} $next\_sibling_{T}(x)$ : The label of the next sibling of the node $x$. \\
\noindent\emph{-} $child_{T}(x, i)$ :  Label of the $i$-th child of the node $x$. \\
\noindent\emph{-} $child\_rank_{T}(x)$ :  Number of siblings left to the node $x$. \\
\noindent\emph{-} $la_{T}(x, i)$ : Label of the level ancestor of node $x$ at depth $i$. \\
\noindent\emph{-} $lca_{T}(x, y)$ : Label of the least common ancestor of node $x$ and $y$.\\ 
\noindent\emph{-} $pre\_rank_{T}(i)$ : The preorder rank of the node in $T$ corresponding to $D_{T}[i]$.  \\
\noindent\emph{-} $pre\_select_T(x)$ :  The first position of node with preorder rank $x$ in $D_T$.

\end{lemma}

We use the following lemma  
to bound the space usage of the
data structures described in Section~\ref{sec:minmax}.

\begin{lemma}
\label{lemma:encoding}
Given two positive integers $a$ and $n$, and a nonnegative integer $k \le n$,
$\lg{n \choose k} + a(n-k) \le n\lg{(2^a+1)}$.
\end{lemma}
\begin{proof}
By raising both sides to the power of $2$, it is enough to prove that
${n \choose k}2^{(a(n-k))} \le (2^a+1)^n$.
We prove the lemma by induction on $n$ and $k$.
In the base case, when $n=1$ and $k=0$, the claim holds since $2^{a} < (2^a+1)$.
Now suppose that ${n' \choose k'}2^{a(n'-k')} \le (2^a+1)^{n'}$ for all $0 < n' \le n$ and $0 \le k' \le k$. Then
\begin{flalign*}
{n+1 \choose k}2^{a(n+1-k)} & = \left({n \choose k} + {n \choose k-1}\right)2^{a(n+1-k)} 
\le 2^a(2^a+1)^n + (2^a+1)^n \\ 
& =  (2^a+1)^{n+1}\mbox{  by induction hypothesis.}  &  
\end{flalign*}

\noindent Also by induction hypothesis,
\begin{flalign*}
& {n \choose k+1}2^{a(n-(k+1))} = \left({n-1 \choose k} + {n-1 \choose k+1}\right)2^{a(n-(k+1))} \le (2^a+1)^{n-1}\left(1+\frac{{n-1 \choose k+1}(2^{a(n-1-k)})}{(2^a+1)^{n-1}}\right)&
\end{flalign*}

\noindent Since ${n-1 \choose k+1}2^{a(n-1-k)} < 2^a(2^a+1)^{n-1} (\because (2^a+1)^{n-1} = \sum_{m=0}^{n-1}{n-1 \choose m}2^{a(n-1-m)})$,
\begin{flalign*}
&(2^a+1)^{n-1}\left(1+\frac{{n-1 \choose k+1}(2^{a(n-1-k)})}{(2^a+1)^{n-1}}\right) < (2^a+1)^{n-1}(1+2^a) = (2^a+1)^n.&
\end{flalign*}
Therefore the above inequality still holds when $n' = n+1$ or $k' =k+1$, which proves the lemma. 
\end{proof}

\subsection{$2d$-Min heap}
\label{subsec:2dminheap}
The $2d$-Min heap~\cite{fh-sjc11} on $A$, denoted by $\MinA{}$, is designed to 
encode the answers of $\RMinQ{}_A(i, j)$ efficiently. 
We can also define the $2d$-Max heap on $A$ ($\MaxA{}$) analogously.
$\MinA{}$ is an ordered labeled tree with $n+1$ nodes labeled with $0 \dots n$. 
Each node in $\MinA{}$ is labeled by its preorder rank and each label corresponds to a position in $A$.
We extend the array $A[1 \dots n]$ to $A[0 \dots n]$ with $A[0] = -\infty$. 
In the labeled tree, the node $x$ denotes the node labeled $x$.
For every vertex $i$, except for the root node, its parent node is (labeled with) $\PSV{}_A(i)$. 


Using the operations in Lemma~\ref{lemma:operation-dfuds}, Fischer and Heun~\cite{fh-sjc11} showed that $\RMinQ{}_A(i, j)$ can be answered in constant time using $D_{\MinA{}}$. If the elements in $A$ are not distinct, $\RMinQ{}_A(i, j)$ returns the $\RRMinQ{}_A(i, j)$. 

Fischer and Heun~\cite{fh-sjc11} also proposed a linear-time stack-based algorithm to construct $D_{\MinA{}}$.
Their algorithm maintains a min-stack consisting of a decreasing sequence of elements from top to the bottom. 
The elements of $A$ are pushed into the min-stack from right to left and before pushing the element $A[i]$, 
all the elements from the stack that are larger than $A[i]$ are popped.
Starting with an empty string, the algorithm constructs a sequence $S$ as described below.
Whenever $k$ elements are popped from the stack and then an element is pushed into the stack, $(^k )$ is prepended to $S$.
Finally, after pushing $A[1]$ into the stack, if the stack contains $m$ elements, then $(^{m+1} )$ is prepended to $S$. 
One can show that this sequence $S$ is the same as the DFUDS sequence $D_{\MinA{}}$.
Analogously, one can construct $D_{\MaxA{}}$ using a similar stack-based algorithm.

\paragraph{\textbf{Colored 2d-Min heap}}
\begin{figure}[htbp]
\begin{center}
    \includegraphics[scale=0.99]{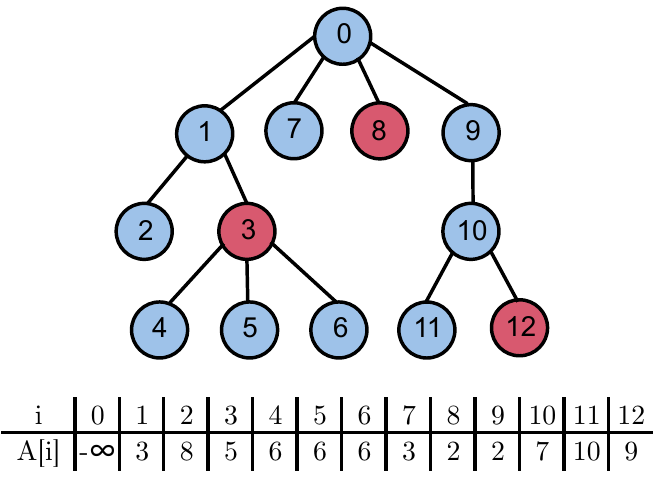}
    \caption{Colored $2d$-Min heap of $A$} \label{fig:color2d}
\end{center}
\end{figure}

From the definition of $2d$-Min heap, it is easy show that $\PSV{}_A(i)$, for $1 \leq i \leq n$, is the label corresponding to the parent of the node labeled $i$ in $\MinA{}$.
Thus, using the encoding of Lemma \ref{lemma:operation-dfuds} using $2n+o(n)$ bits, one can support the $\PSV{}_A(i)$ queries in constant time.
A straightforward way to support $\NSV{}_A(i)$ is to construct the $2d$-Min heap structure for the reverse of the array $A$, and encode it using 
an additional $2n+o(n)$ bits. 
Therefore one can encode all answers of $\PSV{}_A$ and $\NSV{}_A$ using $4n+o(n)$ bits with constant query time.
To reduce this size, Fischer proposed the \textit{colored $2d$-Min heap}~\cite{Fischer11}. 
This has the same structure as normal $2d$-Min heap, and in addition, the vertices are colored either red or blue. 
Suppose there is a parent node $x$ in the colored $2d$-Min heap with its children $x_1 \dots x_k$.
Then for $1 < i \leq k$, node $x_i$ is colored red if  $A[x_i] < A[x_{i-1}]$, and all the other nodes are colored blue (see Figure \ref{fig:color2d}).
We define the operation $\NRS{}(x_i)$ which returns the leftmost red sibling to the right (i.e., next red sibling) of $x_i$.

The following lemma can be used to support $\NSV{}_A(i)$ efficiently using the colored $2d$-Min heap representation.

\begin{lemma}[\cite{Fischer11}]
\label{lemma:c2dheap}
Let $\CMinA{}$ be the colored $2d$-Min heap on $A$. 
Suppose there is a parent node $x$ in $\CMinA{}$ with its children $x_1 \dots x_k$. Then for $1 \leq i \leq k$,\\

$\NSV{}_A(x_i) = \left\{\begin{array}{ll}
 \NRS{}(x_i)  & \textrm{\hspace{1cm} if $\NRS{}(x_i)$ exists,}\\
x_k+subtree\_size(x_k)  &\textrm{\hspace{1.1cm}otherwise.}\\
\end{array} \right.$
\end{lemma}

If all the elements in $A$ are distinct, 
then a $2n+o(n)$-bit encoding of $\MinA{}$ is 
enough to support $\RMinQ{}_A$, $\PSV{}_{A}$ and $\NSV{}_{A}$ with constant query time.
In the general case, Fischer proposed an optimal $2.54n+o(n)$-bit encoding of colored $2d$-Min heap on $A$ using 
TC-encoding~\cite{fm-algo12}.
This encoding also supports two additional operations, namely {\em modified} $child_{\CMinA{}}(x, i)$ and $child\_rank_{\CMinA{}}(x)$,
which answer the $i$-th red child of node $x$ and the number of red siblings to the left of node $x$, respectively, in constant time.
Using these operations, one can also support $\RLMinQ{}_A$ and $\RkMinQ{}_A$ in constant time.

\subsection{Encoding range min-max queries}
\label{subsec:minmax2}
\begin{figure}[htbp]
\begin{center}
    \includegraphics[scale=0.85]{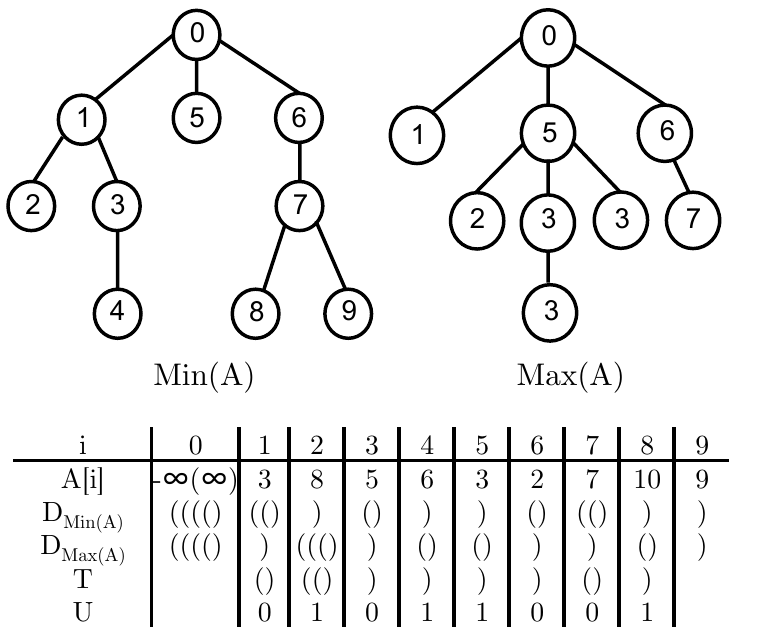}
    \caption{Encoding of $2d$-Min heap and $2d$-Max heap of $A$} \label{fig:minmax}
\end{center}
\end{figure}

One can support both $\RMinQ{}_A$ and $\RMaxQ{}_A$ in constant time 
by encoding both $\MinA{}$ and $\MaxA{}$ separately using $4n+o(n)$ bits.
Gawrychowski and Nicholson~\cite{Gawry14} described an alternate encoding that uses only $3n+o(n)$ 
bits while still supporting the queries in $O(1)$ time.
There are two key observations which are used in obtaining this structure:
\begin{enumerate}
\item If we can access any $\lg{n}$-bit substring of $D_{\MinA{}}$ and $D_{\MaxA{}}$ on $O(1)$ time,
we can still support both queries in $O(1)$ time, using an additional $o(n)$ bits; 
\item 
To generate $D_{\MinA{}}$ and $D_{\MaxA{}}$ using Fischer and Heun's stack-based algorithm,
in each step we push an element into both the min-stack and max-stack, and pop a certain number of elements
from exactly one of the stacks (assuming that $A[i] \neq A[i+1]$, for all $i$, where $1 \le i < n$). 
\end{enumerate}


\begin{algorithm}[t]
\caption{Construction algorithm for $T$ and $U$}
\label{TUalgorithm}
\begin{algorithmic}[1]
\State Initialize $T$ to `)', and $U$ to $\epsilon$.
\State Initialize Min-stack and Max-stack as empty stacks
\State Push $A[n]$ into Min-stack and Max-stack.
\For{$i  := n-1$ to $1$ }
\State counter = 0
\If {$A[i]  < A[i-1]$}
\State Push $A[i]$ into Max-stack
\While {((Min-stack is not empty) \& (Top of Min-stack  $> A[i]$))}
\State Pop Min-stack
\State counter = counter + 1
\EndWhile
\State Push $A[i]$ into Min-stack
\State Prepend $(^{counter-1} ) $ to $T$ and $0$ to $U$
\Else ~ // $A[i]  > A[i-1]$
\State Push $A[i]$ into Min-stack
\While {((Max-stack is not empty) \& (Top of Max-stack  $< A[i]$))}
\State Pop Max-stack
\State counter = counter + 1
\EndWhile
\State Push $A[i]$ into Max-stack
\State Prepend $(^{counter-1} ) $ to $T$ and $1$ to $U$
\EndIf
\EndFor
\end{algorithmic}
\end{algorithm}

Now we describe the overall encoding in~\cite{Gawry14} briefly. 
The structure consists of two bit strings $T$ and $U$ along with various auxiliary structures.
For $1 \le i < n$, if $k$ elements are popped from the min (max)-stack 
when we push $A[i] (1 \leq i < n)$ into both the stacks (from right to left),
we prepend $(^{k-1}~)$ and $0(1)$ to the currently generated $T$ and $U$ respectively.
Initially, when $i=n$, both min and max stacks push `)' so we do not prepend anything to both strings. 
But we can recover it easily because this is the last `)' in common. 
Finally, after pushing $A[1]$ into both the stacks, we pop the remaining elements from them, 
and store the number of these popped elements in min and max stack explicitly using $\lg{n}$ bits.
One can show that the size of $T$ is at most $2n$ bits, and that of $U$ is $n-1$ bits. 
Thus the total space usage is at most $3n$ bits. 
See Algorithm~\ref{TUalgorithm} for the pseudocode, and Figure \ref{fig:minmax} for an example.

To recover any $\lg{n}$-bit substring, $D_{\MinA{}}[d_1 \dots d_{\lg{n}}]$, in constant time 
we construct the following auxiliary structures.
We first divide $D_{\MinA{}}$ into blocks of size $\lg{n}$, and for the starting position 
of each block, store its corresponding position in $T$. 
For this purpose, we construct a bit string  $B_{min}$ of length at most $2n$ 
such that $B_{min}[i]=1$ if and only if $T[i]$ 
corresponds to the start position of the $i$th-block in $D_{\MinA{}}$. 
We encode $B_{min}$ using the representation of 
Lemma~\ref{lemma:rankselect2} which takes $o(n)$ bits since the number of ones in $B_{min}$ is $2n/\lg{n}$.
Then if $d_1$ belongs to the $i$-th block, it is enough to recover the $i$-th 
and the $(i+1)$-st blocks in the worst case.

Now, to recover the $i$-th block of $D_{\MinA{}}$,
we first compute the distance between $i$-th and $(i+1)$-st 1's in $B_{min}$.
If this distance is less than $c\lg{n}$ for some fixed constant $c  > 9$, 
we call it a {\em min-good block}, otherwise, we call it a {\em min-bad block}.
We can recover a min-good block in $D_{\MinA{}}$ in $O(c)$
time using a $o(n)$-bit pre-computed table 
indexed by all possible strings of length $\lg{n}/4$ bits for $T$ 
and $U$ (we can find the position corresponding to the $i$-th block in $U$ in constant time), 
which stores the appropriate $O(\lg{n})$ bits of $D_{\MinA{}}$ obtained from them
(see~\cite{Gawry14} for details).
For min-bad blocks, we store the answers explicitly. 
This takes $(2n/(c\lg{n})) \cdot \lg{n} = 2n/c$ additional bits.
To save this additional space, we store the min-bad blocks in compressed form 
using the property that any min-bad block 
in $D_{\MinA{}}$ and $D_{\MaxA{}}$ cannot overlap more than $4\lg{n}$ bits in $T$,
(since any $2\lg{n}$ consecutive bits in $T$ consist of at least $\lg{n}$ bits 
from either $D_{\MinA{}}$ or $D_{\MaxA{}}$).
So, for $c>9$ we can save more than $\lg{n}$ bits by compressing the 
remaining $(c-4)\lg{n}$ bits in $T$ corresponding to each min-bad block in $D_{\MinA{}}$. 
Thus, we can reconstruct any $\lg n$-bit substring of $D_{\MinA{}}$ (and $D_{\MaxA{}}$) 
in constant time, using a total of $3n+o(n)$ bits. 

We first observe that if there is a position $i$, for $1 \leq i < n$ such that $A[i] = A[i+1]$,
we cannot decode the $`)'$ in $T$ which corresponds to $A[i]$ only using $T$ and $U$ 
since we do not pop any elements from both min and max stacks 
when we push $A[i]$ into both stacks. Gawrychowski and Nicholson~\cite{Gawry14} 
handle this case by defining an ordering between equal elements (for example, by breaking the ties
based on their positions).
But this ordering does not help us in supporting the \PSV{} and \PLV{} queries.
%
We describe how to handle the case when there are repeated (consecutive) elements in $A$,
to answer the \PSV{} and \PLV{} queries. 

Gawrychowski and Nicholson~\cite{Gawry14} also show that any encoding 
that supports both $\RMinQ{}_A$ and $\RMaxQ{}_A$ cannot use less than 
$3n-\Theta(\lg{n})$ bits for sufficiently large $n$ (even if all elements in $A$ are distinct).

\section{Extended DFUDS for colored 2d-Min heap}
\label{sec:exdfuds}

In this section, we describe an encoding of colored $2d$-Min heap on $A$ ($\CMinA{}$) using at most $3n+o(n)$ bits
while supporting $\RMinQ{}_A$, $\RRMinQ{}_A$, $\RLMinQ{}_A$, $\RkMinQ{}_A$, $\PSV{}_A$ and $\NSV{}_A$ 
in constant time. This is done by storing the color information of the nodes using a bit string of length at most $n$, in addition to the DFUDS representation of $\CMinA{}$.
We can also encode the colored $2d$-Max heap in a similar way.
In the worst case, this representation uses more space than the colored $2d$-Min heap encoding of Fischer~\cite{Fischer11},
but the advantage is that it separates the tree encoding from the color information. 
We later describe how to combine the tree encodings of the $2d$-Min heap and $2d$-Max heap, 
and (separately) also combine the color information of the two trees, to reduce the overall space.

Now we describe the main encoding of $\CMinA{}$.
The encoding consists of two parts: $D_{\CMinA{}}$ and $V_{min}$. 
The sequence $D_{\CMinA{}}$ is same as $D_{\MinA{}}$, the DFUDS representation of $\CMinA{}$, which 
takes $2n+o(n)$ bits and supports the operations in Lemma \ref{lemma:operation-dfuds} in constant time.

The bit string $V_{min}$ stores the color information of all nodes in $\CMinA{}$, 
except the nodes which are the leftmost children of their parents (the color of these nodes is always blue),
as follows.
Suppose there are $p$ nodes in $\CMinA{}$, for $1 \leq p \leq n$, which are the leftmost children of their parents.
Then we define the bit string $V_{min}[0 \dots n-p]$ as follows.
For $1 \le i \le n-p$, $V_{min}[i]$ stores $0$ if the color of the node
$$node_{V_{min}}(i) = pre\_rank_{\CMinA{}} (\findclose{}_{D_{\CMinA{}}} (\select{}_{D_{\CMinA{}}}(i+1, `( (' )) +1)$$ 
in $\CMinA{}$ is red, and $1$ otherwise.
This follows from the observation that  if there is an $i, 1 \leq i < 2n-1$ such that
$D_{\CMinA{}}[i] = `('$ and $D_{\CMinA{}}[i+1] = `)'$, then $D_{\CMinA{}}[i+2]$ corresponds to the 
node which is the leftmost child of the node $pre\_rank_{\CMinA{}}(D_{\CMinA{}}[i])$, 
so we skip these nodes by counting the pattern $`(~('$ in $D_{\CMinA{}}$.
Also, we set $V_{min}[0]=1$, which corresponds to the first open parenthesis in $D_{\CMinA{}}$.
For example, for $\CMinA$ in Figure~\ref{fig:color2d}, we store the node 3's color in $V_{min}[4]$.
This is becuase $\select{}_{D_{\CMinA{}}}(5, `( (' ) = 7$,  
$\findclose{}_{D_{\CMinA{}}}(7) + 1 = 11$ and $pre\_rank_{\CMinA{}}(11)=3$ 
(see Figure~\ref{figure:exdefuds}). 
We define the bit string $V_{max}$ in a similar way.

\begin{figure}[htp]

\centering

\begin{center}
\includegraphics[scale=0.99]{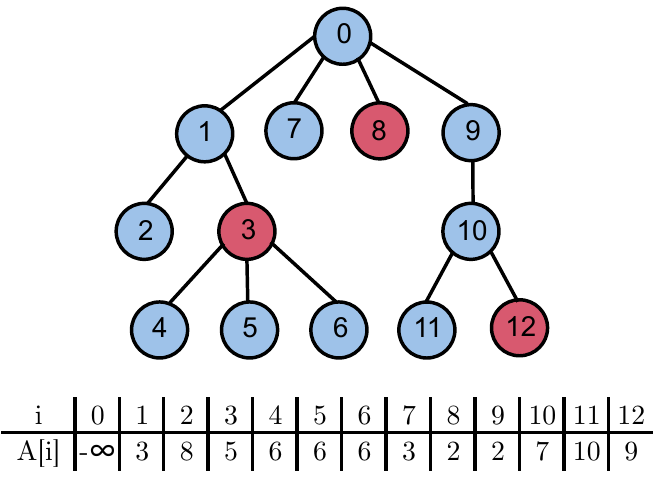}
\end{center}

\scalebox{0.85}{
\begin{tabular}{|c|c|c|c|c|c|c|c|c|c|c|c|c|c|c|c|c|c|c|c|c|c|c|c|c|c|c|}
\hline
$D_{\CMinA{}}$&(&(&(&(&(&)&(&(&)&)&(&(&(&)&)&)&)&)&)&(&)&(&(&)&)&)\\ \hline
$pre\_rank_{\CMinA{}}$&0&0&0&0&0&0&1&1&1&2&3&3&3&3&4&5&6&7&8&9&9&10&10&10&11&12\\
\hline
\end{tabular}
}
\bigskip

\scalebox{0.85}{
\begin{tabular}{|c|c|c|c|c|c|c|c|c|}
\hline
$V_{min}$&1&1&0&1&0&1&1&0 \\\hline
$node_{V_{min}}$&-&9&8&7&3&6&5&12 \\
\hline
\end{tabular}
}
%
%

\bigskip

\scalebox{0.85}{
\begin{tabular}{|c|c|c|c|c|c|c|c|c|c|c|c|c|c|}
\hline
$pre\_select_{\CMinA{}}$&1&7&10&11&15&16&17&18&19&20&22&25&26 \\\hline
$node\_color_{\CMinA{}}$&-&-&-&4&-&6&5&3&2&1&-&-&7\\
\hline
\end{tabular}
}
\caption{$D_{\CMinA{}}$, $pre\_rank_{\CMinA{}}$, $V_{min}[i]$, $node_{V_{min}}$, $pre\_select_{\CMinA{}}$ and $node\_color_{\CMinA{}}$ 
for colored $2d$-Min heap}\label{figure:exdefuds}
\end{figure}

The following lemma
shows that encoding 
$\MinA{}$ and $V_{min}$ separately, using at most $3n+o(n)$ bits, has the same functionality as the $\CMinA{}$ 
encoding of Fischer~\cite{Fischer11}, which only takes $2.54n+o(n)$ bits.

\begin{lemma}
\label{lemma:extenddfuds}
For an array $A[1 \dots n]$ of length $n$, there is an encoding for $A$ which takes at most $3n+o(n)$ bits
and supports $\RMinQ{}_A$, $\RRMinQ{}_A$, $\RLMinQ{}_A$, $\RkMinQ{}_A$, $\PSV{}_A$ and $\NSV{}_A$ 
in constant time. 
\end{lemma}
\begin{proof}
The encoding consists of the $2n+o(n)$-bit encoding of $\MinA{}$ encoded using structure of~Lemma \ref{lemma:operation-dfuds},
together with the bit string $V_{min}$ that stores the color information of the nodes in $\CMinA{}$. 
We use a $o(n)$-bit auxiliary structure to support the \rank{}/\select{} queries on $V_{min}$ in constant time.
Since the size of $V_{min}$ is at most $n$ bits, the total space of the encoding is at most $3n+o(n)$ bits.

To define the correspondence between the nodes in $\CMinA{}$ and the positions in the bit string $V_{min}$, we define the following operation.
For $0 \le i \le n$, we define $node\_color_{\CMinA{}}(i)$ as the position of 
 $V_{min}$ that stores the color of the node $i$ in $\CMinA{}$. 
This can be computed in constant time, using $o(n)$ bits, by
$$node\_color_{\CMinA{}}(i) = \left\{\begin{array}{ll}
 \textrm{undefined} & \textrm{\hspace{0.3cm}if $child\_rank_{\CMinA{}}(i) = 0$}\\
 \rank{}_{D_{\CMinA{}}}( c , `((' ) - 1 & \textrm{\hspace{0.3cm}otherwise}\\
\end{array} \right.$$
where $c = \findopen_{D_{\CMinA{}}}(pre\_select_{\CMinA{}}(i)-1)$
(note that $node\_color_{\CMinA{}}$ is the inverse operation of $node_{V_{min}}$, i.e, 
if $node_{V_{min}}(k) = i$, then $node\_color_{\CMinA{}}(i) = k$).

Now we describe how to support the queries in constant time.
Fischer and Heun~\cite{fh-sjc11} showed that $\RMinQ{}_A(i, j)$ can be answered in constant time using $D_{\CMinA{}}$.
In fact, they return the position $\RRMinQ{}_A(i, j)$ as the answer to $\RMinQ{}_A(i, j)$.
Also, as mentioned earlier, $\PSV{}_A(i)$ = $parent_{\CMinA{}}(i)$, and hence can be answered in constant time. 
%
Therefore, it is enough to describe how to find $\RLMinQ{}_A(i, j)$, $\RkMinQ{}_A(i, j)$ and $\NSV{}_A(i)$ in constant time.
\\\\
\noindent{\textbf{$\boldsymbol{\RLMinQ{}_{A}(i, j)$}}}: 
As shown by Fischer and Huen~\cite{fh-sjc11}, all corresponding values of 
left siblings of the node $\RRMinQ{}_{A}(i, j)$ in $A$ 
are at least $A[\RRMinQ{}_{A}(i, j)]$ (i.e., the values of the siblings are in the non-increasing order, from left to right).
Also, for a child node $m$ of any of the left siblings of the node $\RRMinQ{}_{A}(i, j)$,  $A[m] > A[\RRMinQ{}_{A}(i, j)]$. 
Therefore, the position $\RLMinQ{}_{A}(i, j)$ corresponds to one of the left siblings 
of the node whose position corresponds to $\RRMinQ{}_{A}(i, j)$.

We first check whether the color of the node $\RRMinQ{}_{A}(i, j)$ is red or not using $V_{min}$.
If $$V_{min}[node\_color_{\CMinA{}}(\RRMinQ{}_{A}(i, j))] = 0$$ 
then $\RLMinQ{}_{A}(i, j)= \RRMinQ{}_{A}(i, j)$.
If not, we find the node $\leftmost(i, j)$ which is the leftmost sibling 
of the node $\RRMinQ{}_{A}(i, j)$ between the nodes in $[i \dots j]$.
$\leftmost(i, j)$ can be found in constant time by computing the depth of node $i$ and 
comparing this value with $d_{right}$, the depth of the node $\RRMinQ{}_{A}(i, j)$. More specifically,
$$\leftmost(i, j) = \left\{\begin{array}{ll}
 i & \textrm{\hspace{0.3cm}if $depth_{\CMinA}(i)=d_{right}$.}\\
 next\_sibling_{\CMinA{}}(la_{\CMinA{}}(i, d_{right})) & \textrm{\hspace{0.3cm}otherwise.}\\
\end{array} \right.$$
In the next step, find the leftmost blue sibling $n_v$
such that there is no red sibling between $n_v$ and $\RRMinQ{}_{A}(i, j)$. 
This can be found in constant time by first finding the index $v$ using the equation
$$v = \select_{V_{min}} (\rank_{V_{min}}(node\_color_{\CMinA{}}( \RRMinQ{}_{A}(i, j)), 0)+1, 0)-1$$
and then finding the node $n_v$ using $n_v = node_{V_{min}}(v)$.
%
If $child\_rank_{\CMinA{}}(n_v) \le child\_rank_{\CMinA{}} (\leftmost(i, j))$ or $child\_rank_{\CMinA{}}(n_v) = 1$ 
(this is the case that $\leftmost(i, j)$ can be the the lestmost sibling), then
$\RLMinQ{}_{A}(i, j) = \leftmost(i, j)$. Otherwise, $\RLMinQ{}_{A}(i, j) = n_v$.
\\\\
\noindent{\textbf{ $\boldsymbol{\RkMinQ{}_{A}(i, j)$}}}: 
This query can be answered in constant time by returning the $k$-th sibling (in the left-to-right order)
of $\RLMinQ{}_{A}(i, j)$, if it exists. 
More formally, if $child\_rank_{\CMinA{}}(\RRMinQ{}_{A}(i, j)) - child\_rank_{\CMinA{}}(\RLMinQ{}_{A}(i, j))$ is at least $k-1$,
then $\RkMinQ{}_{A}(i, j)$ exists; and in this case,
$\RkMinQ{}_{A}(i, j)$ can be computed in constant time by computing
$$child_{\CMinA{}}(parent_{\CMinA{}}(\RRMinQ{}_{A}(i, j)), \RLMinQ{}_{A}(i, j)+k-1).$$
%
\\
\noindent{\textbf{ $\boldsymbol{\NSV{}_{A}(i)$}}}: 
%
%
%
By Lemma~\ref{lemma:c2dheap}, it is enough to show how to support $\NRS{}(i)$
in constant time (note that we can support $subtree\_size$ in constant time using Lemma~\ref{lemma:operation-dfuds}).
If node $i$ is the rightmost sibling, then $\NRS{}(i)$ does not exist.
Otherwise we define $v'$ as 
$\select{}_{V_{min}}(\rank{}_{V_{min}}(node\_color_{\CMinA{}} (next\_sibling(i)) , 0), 0)$.
Let $n_{v'} = node_{V_{min}}(v')$.
If the parent of $n_{v'}$ is same as the parent of $i$, then
 $\NRS{}(i) = n_{v'}$; otherwise $\NRS{}(i)$ does not exist.
Finally, if $\NRS{}(i)$ does not exist, 
we compute the node $r$ which is the rightmost sibling of the node $i$ can be found by 
$$child_{\CMinA{}}(parent_{\CMinA{}}(i), degree_{\CMinA{}}(parent_{\CMinA{}}(i))-1).$$ 
Then $\NSV{}_{A}(i) = r + subtree\_size_{\CMinA{}}(r)$.
All these operations can be done in constant time.
\end{proof}

\section{Encoding colored 2d-Min and 2d-Max heaps}
\label{sec:minmax}
In this section, we describe our encodings for supporting various subsets of operations, proving the results stated in Theorem~\ref{theorem:minmax}.
As mentioned in Section~\ref{subsec:2dminheap}, the TC-encoding of the colored $2d$-Min heap of Fischer~\cite{Fischer11}
can answer $\RMinQ{}_A$, $\RRMinQ{}_A$, $\PSV{}_A$ and $\NSV{}_A$ queries
in $O(1)$ time, using $2.54n+o(n)$ bits.
The following lemma shows that we can also support the queries $\RLMinQ{}_A$ and $\RkMinQ{}_A$ 
using the same structure.
\begin{lemma}
\label{lemma:tcencoding}
For an array $A[1 \dots n]$ of length $n$, 
$\RLMinQ{}_A$, $\RkMinQ{}_A$ can be answered in constant time by the
TC-encoding of colored $2d$-Min heap.
\end{lemma}
\begin{proof}
Fischer~\cite{Fischer11} defined two operations, which are modifications of the $child$ and $child\_rank$, as follows:
\begin{itemize}
\item $mchild_{\CMinA{}}(x, i)$ - returns the $i$-th red child of node $x$ in \CMinA{}, and
\item $mchild\_rank_{\CMinA{}}(x)$ - returns the number of red siblings to the left of node $x$ in \CMinA{}.
\end{itemize}
He showed that the TC-encoding of the colored $2d$-Min heap 
can support $mchild_{\CMinA{}}(x, i)$ and $mchild\_rank_{\CMinA{}}(x)$ in constant time. 
Also, since the TC-encoding supports $depth_{\CMinA{}}$, $next\_sibling_{\CMinA{}}$,
$la_{\CMinA{}}$, $child_{\CMinA{}}$ and $child\_rank_{\CMinA{}}$ in constant time on ordinal trees~\cite{hms-icalp07},
we can support $\leftmost(i, j)$ (defined in the proof of the Lemma~\ref{lemma:extenddfuds}) in constant time.
For answering $\RLMinQ{}_A(i, j)$, we first find the previous red sibling $l$ 
of $\RRMinQ{}_A(i, j)$ using $mchild_{\CMinA{}}$ and $mchild\_rank_{\CMinA{}}$. 
If such a node $l$ exists, we compare the child ranks of $next\_sibling_{\CMinA}(l)$ 
and $\leftmost(i, j)$, and return the node with the larger rank value as the answer.
$\RkMinQ{}_A(i, j)$ can be answered by returning the $k$-th sibling (in the left-to-right order) of $\RLMinQ{}_A(i, j)$
using $child_{\CMinA{}}$ and $child\_rank_{\CMinA{}}$, if it exists.
\end{proof}
By storing a similar TC-encoding of colored $2d$-Max heap, in addition to the 
structure of Lemma~\ref{lemma:tcencoding}, we can support all the operations 
mentioned in Theorem~\ref{theorem:minmax}(c) in $O(1)$ time. This uses a
total space of $5.08n+o(n)$ bits. We now describe alternative encodings to reduce the 
overall space usage.

More specifically, we show that a combined encoding of $D_{\CMinA{}}$ and $D_{\CMaxA{}}$, using at most $3.17n+o(n)$ bits, can be used to answer 
$\RMinQ{}_A$, $\RMaxQ{}_A$, $\RRMinQ{}_A$, $\RRMaxQ{}_A$, $\PSV{}_A$, and $\PLV{}_A$ queries (Theorem~\ref{theorem:minmax}(a)).
To support the queries in constant time, we use a less space-efficient data structure that encodes the same structures, using at most $3.322n+o(n)$ bits 
(Theorem~\ref{theorem:minmax}(b)). 
Similarly, a combined encoding of $D_{\CMinA{}}$, $D_{\CMaxA{}}$, $V_{min}$ and $V_{max}$ using at most $4.088n+o(n)$ bits 
can be used to answer $\RLMinQ{}_A$, $\RkMinQ{}_A$, $\NSV{}_A$, $\RLMaxQ{}_A$, $\RkMaxQ{}_A$, and $\NLV{}_A$ 
queries in addition (Theorem~\ref{theorem:minmax}(c)). 
Again, to support the queries in constant time, we design a 
less space-efficient data structure using at most $4.58n+o(n)$ bits (Theorem~\ref{theorem:minmax}(d)).

In the following, we first describe the data structure of Theorem~\ref{theorem:minmax}(b) 
followed by the structure for Theorem~\ref{theorem:minmax}(d). 
Next we describe the encodings of Theorem~\ref{theorem:minmax}(a) and Theorem~\ref{theorem:minmax}(c).

\subsection{Combined data structure for $D_{\CMinA{}}$ and $D_{\CMaxA{}}$} 
\label{subsubdd}

As mentioned in Section~\ref{subsec:minmax2}, the encoding of Gawrychowski and Nicholson~\cite{Gawry14} 
consists of two bit strings $U$ and $T$ of total length at most $3n$, along with the encodings of $B_{min}$, $B_{max}$ and 
a few additional auxiliary structures of total size $o(n)$ bits. In this section, we denote this encoding by $E$.
To encode the DFUDS sequences of $\CMinA{}$ and $\CMaxA{}$ in a single structure, we use $E$
with some modifications, which we denote by $E'$.
As described in Section~\ref{subsec:minmax2}, encoding scheme of Gawrychowski and Nicholson 
cannot be used (as it is) to support the $\PSV{}$ and $\PLV{}$ queries 
if there is a position $i$, for $1 \leq i < n$ such that $A[i] = A[i+1]$.
To support these queries, we define an additional bit string
$C[1 \dots n]$ such that $C[1] = 0$, and for $1 < i \leq n$, $C[i] = 1$ iff $A[i-1] = A[i]$. 
If the bit string $C$ has $k$ ones in it, then
we represent $C$ using $\lg{n \choose k} + o(n)$ bits
while supporting \rank{}, \select{} queries and decoding any $\lg{n}$ consecutive bits in $C$ in constant time,
using Lemma~\ref{lemma:rankselect2}.
We also define a new array $A'[0 \dots n-k]$ 
by setting $A'[0] = A[0]$, and for $0 < i \le n-k$, $A'[i] = A[\select_{C}(i, 0)]$.
(Note that $A'$ has no consecutive repeated elements.)
In addition, we define another sequence $D'_{\CMinA{}}$ of size $2n-k$ as follows.
Suppose $D_{\CMinAp{}} = (^{\delta_1}~) \dots (^{\delta_{n-k}}~)$, for some $ 0 \leq \delta_1 \dots \delta_n \leq n-k$,
then we set $D'_{\CMinA{}} = (^{\delta_1+\epsilon_1}~) \dots (^{\delta_{n-k}+\epsilon_{n-k}}~)$, where $\delta_i+\epsilon_i$ is the number of elements popped 
when $A[i]$ is pushed into the min-stack of $A$, for $ 1 \leq i \leq n-k$.
(Analogously, we define $D'_{\CMaxA{}}$.)

The encoding $E'$ defined on $A$ consists of two bit strings $U'$ and $T'$, 
along with $C$, $B'_{min}$, $B'_{max}$ and additional auxiliary structures (as in $E$). 
Let $U$ and $T$ be the bit strings in $E$ defined on $A'$. 
Then $U'$ is same as $U$ in $E$, and size of $U'$ is $n-k-1$ bits. 
To obtain $T'$, we add some additional open parentheses to $T$ as follows.
Suppose $T = (^{\delta_1}~) (^{\delta_2}~) \dots (^{\delta_{n-k}}~)$, where $ 0 \leq \delta_i \leq n-k$ for $1 \le i \le n-k$.
 Then $T' = (^{\delta_1+\epsilon_1}~) \dots (^{\delta_{n-k}+\epsilon_{n-k}}~)$, 
where $\delta_i+\epsilon_i$ is the number of elements are popped when $A[i]$ is pushed into the min or max stack of $A$, for $ 1 \leq i \leq n-k$
(see Figure \ref{fig:minmax2} for an example). 
Since the length of $T$ is at most $2(n-k)$, and $|T'|-|T| = \sum_{i=1}^{n-k} \epsilon_i \le k$, 
the size of $T'$ is at most $2n-k$ bits.
The encodings of $B'_{min}$ and $B'_{max}$ are defined on $D'_{\CMinA{}}$, $D'_{\CMaxA{}}$ and $T'$, 
 similar to $B_{min}$ and $B_{max}$ in $E$.
The total size of the encodings of the modified $B'_{min}$ and $B'_{max}$ is $o(n)$ bits. 
All the other auxiliary structures use $o(n)$ bits.
Although we use $E'$ instead of $E$, we can use the decoding algorithm in $E$ without any modifications because
all the properties used in the algorithm still hold even though $T'$ has additional open parentheses compared to $T$.
Therefore from $E'$ we can reconstruct any 
$\lg{n}$ consecutive bits of $D'_{\CMinA{}}$ or $D'_{\CMaxA{}}$ in constant time,
and thus we can support \rank{} and \select{} on these 
strings in constant time with $o(n)$ additional structures by Lemma \ref{lemma:rankselect2}.

\begin{figure}[htbp]
\begin{center}
    \includegraphics[scale=0.92]{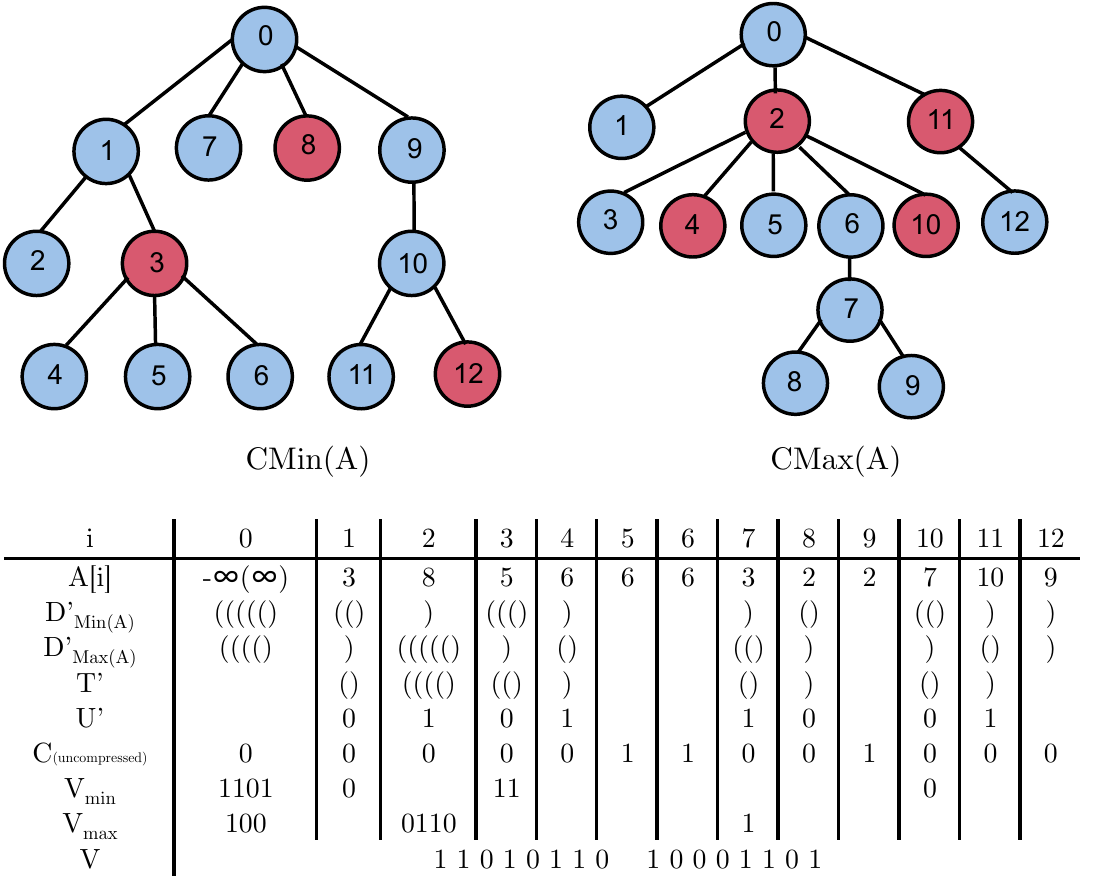}
    \caption{Data structure combining the colored $2d$-Min heap and colored $2d$-Max heap of $A$.  
    $C$ is represented in uncompressed form.} \label{fig:minmax2}
\end{center}
\end{figure}

\subsubsection{Decoding $D_{\CMinA{}}$ and $D_{\CMaxA{}}$}
We use the following auxiliary structures to decode $D_{\CMinA{}}$ from $D'_{\CMinA{}}$ and $C$. 
For this, we first define a correspondence between $D_{\CMinA{}}$ and $D'_{\CMinA{}}$ as follows.
Note that both $D_{\CMinA{}}$ and $D'_{\CMinA{}}$ have the same number of open parentheses,
but $D'_{\CMinA{}}$ has fewer close parentheses than $D_{\CMinA{}}$.
The $i$th open parenthesis in $D_{\CMinA{}}$ corresponds to the $i$th open parenthesis in $D'_{\CMinA{}}$.
Suppose there are $\ell$ and $\ell'$ ($\le \ell$) close parentheses between the $i$th and the $(i+1)$st open parentheses 
in $D_{\CMinA{}}$ and $D'_{\CMinA{}}$, respectively. Then the last $\ell'$ close parentheses in $D_{\CMinA{}}$
correspond, in that order, to the $\ell'$ close parentheses in $D'_{\CMinA{}}$; the remaining close parentheses in 
$D_{\CMinA{}}$ do not have a corresponding position in $D'_{\CMinA{}}$.

We construct three bit strings $P_{min}$, $Q_{min}$ and $R_{min}$ of lengths $2n-k$, $\lceil 2n/\lg{n}\rceil$ and $\lceil 2n/\lg{n}\rceil$, respectively, as follows.
For $1 \leq i \leq \lceil 2n/\lg{n}\rceil$, if the position $i\lg{n}$  in $D_{\CMinA{}}$ has its corresponding position $j$ in $D'_{\CMinA{}}$, 
then we set $P_{min}[j] = 1$, $Q_{min}[i]=0 $ and $R_{min}[i] = 0$.
If position $i\lg{n}$ in $D_{\CMinA{}}$ has no corresponding position in $D'_{\CMinA{}}$ 
but for some $k_i$ where $1 \le k_i < \lg n$, suppose there is a leftmost position $q = i\lg{n}+k_i$
which has its corresponding position $j$ in $D'_{\CMinA{}}$.
Then we set $P_{min}[j] = 1$, $Q_{min}[i]=1 $ and $R_{min}[i] = 0$.
Finally, if all positions between $i\lg{n}$ and $(i+1)\lg{n}$ in $D_{\CMinA{}}$ have no 
corresponding position in $D'_{\CMinA{}}$, then we set $Q_{min}[i]=1 $ and $R_{min}[i] = 1$.
In remaining positions for $P_{min}$, $Q_{min}$ and $R_{min}$, we set their values as $0$.
We also store the values, $k_i$ explicitly, for $1 \leq i \leq \lceil 2n/\lg{n}\rceil$, whenever they are defined
(as in the second case). Since $k_i < \lg{n}$, we can store all the $k_i$ values 
explicitly using at most  $2n\lg{\lg{n}}/\lg{n} = o(n)$ bits. 

Since the bit strings $P_{min}$, $Q_{min}$ and $R_{min}$ have at most $2n/\lg{n}$ 1's each,
they can be represented using the structure of Lemma \ref{lemma:rankselect2}, 
taking $o(n)$ bits while supporting \rank{} and \select{} queries in constant time.
We define $P_{max}$, $Q_{max}$, $R_{max}$ in the same way, and represent them analogously.

In addition to these bit strings, we construct two pre-computed tables.
In the rest of this section, we refer to the parenthesis strings (such as 
$D_{\CMinA{}}$ and $D'_{\CMinA{}}$) also as bit strings.
To describe these tables, we first define two functions $f$ and $f'$,
each of which takes two bit strings $s$ and $c$ as parameters, and returns
a bit string of length at most $|s|+|c|$, as follows.
%
\\

\begin{minipage}{0.0\linewidth}  
$\left\{\begin{array}{lllll}
f(s, \epsilon) = s \\
f(\epsilon, c) = \epsilon \\
f(s, 1 \cdot c_1) = ) \cdot f(s, c_1) \\
f((^{\delta} ~ ) \cdot s_1, 0  \cdot c_1) = (^{\delta} \cdot ) \cdot f(s_1, c_1) 
\end{array} \right.$
\end{minipage} 
\hspace{5.8cm} 
\begin{minipage}{0\linewidth} 
$\left\{\begin{array}{lllll}
f'(s, \epsilon) = s \\
f'(\epsilon, c) = \epsilon \\
f'(s,  c_1  \cdot 1) = f'(s, c_1) \cdot ) \\
f'(s_1 \cdot (^{\delta} ~ ) , c_1  \cdot 0) = f'(s_1, c_1) \cdot (^{\delta} \cdot )
\end{array} \right.$ 
\end{minipage}
\\

One can easily show that if $s$ is a substring of $D'_{\CMinA{}}$ and $c$ is a substring of $C$
whose starting (ending) position corresponds to the starting (ending) position in $s$,
then $f(s, c)$ ($f'(s, c)$) returns the substring of $D_{\CMinA{}}$ whose 
starting (ending) position corresponds to the starting (ending) position in $s$,

We construct a pre-computed table $T_f$ that, for each possible choice of bit strings $s$ and $c$ 
of length $(1/4)\lg{n}$, stores the bit string $f(s, c)$.
These pre-computed tables can be used to decode a substring of 
$D_{\CMinA{}}$ given a substring of $D'_{\CMinA{}}$ (denoted $s$) and a 
substring of $C$ whose bits correspond to $s$.
The total space usage of $T_f$ 
is $2^{(1/4)\lg{n}} \cdot 2^{(1/4)\lg{n}} \cdot ((1/2)\lg{n}) = o(n)$ bits.
We can also construct  $T_{f'}$ defined analogous to $T_f$ using $o(n)$ bits.

Now we describe how to decode $\lg{n}$ consecutive bits of $D_{\CMinA{}}$ in 
constant time. 
(We can decode $\lg{n}$ consecutive bits of $D_{\CMaxA{}}$ in a similar way.)
Suppose we divide $D_{\CMinA{}}$ into blocks of size $\lg{n}$.
As described in Section \ref{subsec:minmax2}, it is enough to show that 
for $1 \leq i \leq \lceil 2n/\lg{n}\rceil$,
we can decode $i$-th block of $D_{\CMinA{}}$ in constant time.
First, we check the value of the $R_{min}[i]$. 
If $R_{min}[i] = 1$, then the $i$-th block in $D_{\CMinA{}}$ consists of a sequence of $\lg{n}$ consecutive close parentheses.
Otherwise, there are two cases depending on the value of $Q_{min}[i]$.
We compute the position $p$ which is a position in $D'_{\CMinA{}}$ (it's exact correspondence 
 in $D_{\CMinA{}}$ depends on the value of the bit $Q_{min}[i]$), and then compute
the position $c_p$ in $C$ which corresponds to $p$ in $D'_{\CMinA{}}$, using the following equations:
$$p = \select{}_{P_{min}}(i - \rank{}_{R_{min}}(i, 1) ,1)$$
$$c_p = \left\{\begin{array}{ll}
 \select{}_{C}( \rank{}_{D'_{\CMinA{}}}(p,~ ')' ), 0) & \textrm{\hspace{0.3cm} if $D'_{\CMinA{}}[p] = ')'$}\\
 \select{}_{C}(\rank{}_{D'_{\CMinA{}}}(p,~ ')' )+1, 0) & \textrm{\hspace{0.3cm}otherwise}\\
\end{array} \right.$$
\\\\
\noindent{\textbf{Case (1) }$\boldsymbol{Q_{min}[i] = 0$}.}
In this case, we take the $\lg{n}$ consecutive bits of $D'_{\CMinA{}}$ starting from $p$, and the 
$\lg{n}$ consecutive bits of $C$ starting from the position $c_p$
(padding at the end with zeros if necessary). 
Using these bit strings, we can decode the  $i$-block in $D_{\CMinA{}}$ by looking up $T_f$ 
with these substrings (a constant number of times, until the pre-computed table generates 
the required $\lg{n}$ bits).
Since the position $p$ corresponds to the starting position of the $i$-th block in $D_{\CMinA{}}$ in this case, 
we can decode the $i$-th block of $D_{\CMinA{}}$ in constant time.
\\\\
\noindent{\textbf{Case (2) }$\boldsymbol{Q_{min}[i]= 1$}.}
First we decode $\lg{n}$ consecutive bits of $D_{\CMinA{}}$ whose starting position 
corresponds to the position $p$ using the same procedure as in Case (1).
Let $S_1$ be this bit string.
Next, we take the $\lg{n}$ consecutive bits of $D'_{\CMinA{}}$ ending with position $p$, 
and the $\lg n$ consecutive bits of $C$ ending with position $c_p$
(padding at the beginning with zeros if necessary). 
Then we can decode the $\lg{n}$ consecutive bits of 
$D_{\CMinA{}}$ whose ending position corresponds to the $p$
by looking up $T_{f'}$ (a constant number of times) with these substrings. 
Let $S_2$ be this bit string.
By concatenating $S_1$ and $S_2$, we obtain a 
$2\lg{n}$-bit substring of $D_{\CMinA{}}$
which contains the starting position of the $i$-th block of $D_{\CMinA{}}$
(since the starting position of the $i$-th block in $D_{\CMinA{}}$, and the 
position which corresponds to $p$ differ by at most $\lg{n}$).
Finally, we can obtain the $i$-th block in $D_{\CMinA{}}$ by skipping the first 
$\lg{n}-k_i$ bits in $S_1 \cdot S_2$, and taking $\lg{n}$ consecutive bits from there. 

From the encoding described above, we can decode any $\lg{n}$ 
consecutive bits of $D_{\CMinA{}}$ and $D_{\CMaxA{}}$ in constant time.
Therefore by Lemma~\ref{lemma:extenddfuds},
we can answer all queries supported by $\CMinA{}$ and $\CMaxA{}$ 
(without using the color information) in constant time.
If there are $k$ elements such that $A[i-1] = A[i]$ for $1 \leq i \leq n$, then the size of $C$ is 
$\lg{n \choose k} + o(n)$ bits, and the size of $E'$ on $A$ is
$3n-2k+o(n)$ bits. All other auxiliary  bit strings and tables take $o(n)$ bits.
Therefore, by the Lemma~\ref{lemma:encoding}, we can encode $A$ 
using $3n-2k+\lg{n \choose k}+o(n) \le ((1+\lg{5})n+o(n) < 3.322n+o(n)$
bits. Also, this encoding supports the
queries in Theorem~\ref{theorem:minmax}(b) 
(namely $\RMinQ{}_A$, $\RMaxQ{}_A$, $\RRMinQ{}_A$, 
$\RRMaxQ{}_A$,  $\PSV{}_A$ and $\PLV{}_A$,
which do not need the color information) in constant time. 
This proves Theorem~\ref{theorem:minmax}(b).

Note that if $k=0$ (i.e, there are no consecutive equal elements), $E'$ on $A$ is same as $E$ on $A$.
Therefore, we can support all the queries in Theorem~\ref{theorem:minmax}(b) using $3n+o(n)$ bits
with constant query time.

\subsubsection{Encoding $V_{min}$ and $V_{max}$}
We simply concatenate $V_{max}$ and $V_{min}$ on $A$ and store it as bitstring $V$, and store the 
length of $V_{min}$ using $\lg{n}$ bits (see $V$ in Figure \ref{fig:minmax2}).
If there are $k$ elements such that $A[i-1] = A[i]$ for $1 \leq i \leq n$, 
Fischer and Heun's stack based algorithm~\cite{fh-sjc11}
does not pop any elements from both stacks when these $k$ elements and $A[n]$ are pushed into them.
Before pushing any of the remaining elements into the min- and max-stacks, 
we pop some elements from exactly one of the stacks.
Also after pushing $A[1]$ into both the stacks, we pop the remaining 
elements from the stacks in the final step.
Suppose
the $n$ elements are popped from the min-stack during $p$ runs of pop operations.
Then, it is easy show that the elements are popped from the max-stack during 
$n-k-p$ runs of pop operations.
Also, $p$ ($n-k-p$) is the number of leftmost children in $\CMinA$ ($\CMaxA$) 
since each run of pop 
operations generates exactly one open parenthesis whose 
matched closing parenthesis corresponds to the leftmost child in $\CMinA{}$ ($\CMaxA$). 
As described in Section~\ref{sec:exdfuds}, 
the size of $V_{min}$  is $n-p+1$ bits, and that of $V_{max}$ is $p+k+1$ bits. 
Thus, the total size of $V$ is $n+k+2$ bits.

Therefore, we can decode any $\lg{n}$-bit substring of $V_{min}$ or $V_{max}$ 
in constant time using $V$ and the length of $V_{min}$.
By combining these structures with the encoding of Theorem~\ref{theorem:minmax}(b), 
we can support the queries in Theorem~\ref{theorem:minmax}(d) 
(namely, the queries $\RMinQ{}_A$, $\RRMinQ{}_A$, $\RLMinQ{}_A$, $\RkMinQ{}_A$, $\PSV{}_A$, $\NSV{}_A$ 
$\RMaxQ{}_A$, $\RRMaxQ{}_A$, $\RLMaxQ{}_A$, $\RkMaxQ{}_A$, $\PLV{}_A$ and $\NLV{}_A$)
in constant time.
By Lemma~\ref{lemma:encoding}, 
the total space of these structures is $4n-k+\lg{n \choose k}+o(n) \le ((3+\lg{3})n+o(n) < 4.585n+o(n)$ bits.
This proves Theorem~\ref{theorem:minmax}(d).

Note that if $k=0$ (i.e, there are no consecutive equal elements), 
$E'$ on $A$ is same as $E$ on $A$, and the size of $V$ is $n+2$ bits.
Therefore we can support all the queries in Theorem~\ref{theorem:minmax}(d) using $4n+o(n)$ bits
with constant query time.

\subsection{Encoding colored $2d$-Min and $2d$-Max heaps using less space}
In this section, we give new encodings that prove Theorem~\ref{theorem:minmax}(a) and Theorem~\ref{theorem:minmax}(c), 
which use less space but take more query time than the previous encodings.
To prove Theorem~\ref{theorem:minmax}(a), we maintain the encoding $E'$ on $A$, with the modification
that instead of $T'$ (which takes at most $2n -k$ bits), we store the bit string $T$ (which takes at most $2(n-k)$ bits) 
which is obtained by constructing the encoding $E$ on $A'$.
Note that $f(s,c)$ is well-defined when $s$ and $c$ are substrings of $D_{\CMinAp{}}$  and C, respectively.
If there are $k$ elements such that $A[i-1] = A[i]$ for $1 \leq i \leq n$, then the total size of 
the encoding is at most 
$3(n-k)+\lg{n \choose k}+o(n) \le n\lg{9}+o(n) < 3.17n+o(n)$ bits. 
If we can reconstruct the sequences $D_{\CMinA{}}$ and $D_{\CMaxA{}}$, by 
Lemma~\ref{lemma:extenddfuds}, we can support all the required queries.
We now describe how to decode the entire $D_{\CMinA{}}$ 
using this encoding. 
(Decoding $D_{\CMaxA{}}$ can be done analogously.)

Once we decode the sequence $D_{\CMinAp{}}$, 
we reconstruct the sequence $D_{\CMinA{}}$ by scanning the sequences 
$D_{\CMinAp{}}$ and $C$ from left to right, and using the lookup table $T_f$.
Note that $f(D_{\CMinAp{}}, C)$ looses some open parentheses in $D_{\CMinA{}}$ 
whose matched close parentheses are not in $D_{\CMinAp{}}$ but in $f(D_{\CMinAp{}}, C)$.
So when we add $m$ consecutive close parentheses from the $r$-th position 
in $D_{\CMinAp{}}$ in decoding with $T_f$, we add $m$ more open parentheses 
before the position $pos = \findopen_{D_\MinAp{}}(r-1)$. This proves Theorem~\ref{theorem:minmax}(a).

To prove Theorem~\ref{theorem:minmax}(c), we combine the concatenated sequence of $V_{min}$ and $V_{max}$ on $A'$ 
whose total size is $n-k+2$ bits to the above encoding. Then we can reconstruct $V_{min}$ on $A$
by adding $m$ extra $1$'s before 
$V_{min}[\rank{}_{D_\MinAp{}}(pos, `((' )]$ when $m$ consecutive close parentheses are added from the $r$-th position 
in $D_{\CMinAp{}}$ while decoding with $T_f$.
(Reconstructing $V_{max}$ on $A$ can be done in a similar way.)
The space usage of this encoding is $4(n-k)+\lg{n \choose k}+o(n) \le n\lg{17}+o(n) < 4.088n+o(n)$ bits.
This proves Theorem~\ref{theorem:minmax}(c).

\section{Conclusion}
\label{sec:concl}
We obtained space-efficient encodings that support a large set of range and previous/next smaller/larger value queries.
The encodings that support the queries in constant time take more space than the ones that do not support the queries in constant time.
\\\\
We conclude with the following open problems.
\begin{itemize}
\item
Can we support the queries in the Theorem~\ref{theorem:minmax}(c) in $O(1)$ time using at most $4.088n+o(n)$ bits?
\item
As described in Section~\ref{sec:prelim}, 
Gawrychowski and Nicholson~\cite{Gawry14} show that any encoding that supports 
both $\RMinQ{}_A$ and $\RMaxQ{}_A$ requires at least $3n-\Theta(\lg{n})$ bits. 
Can we obtain an improved lower bound in the case when we need to support the queries in 
Theorem~\ref{theorem:minmax}(a)?
\item Can we prove a lower bound that is strictly more than $3n$ bits for any 
encoding that supports the queries in Theorem~\ref{theorem:minmax}(c)?
\end{itemize}

\noindent
{\bf Acknowledgments. } We thank Ian Munro and Meng He for helpful discussions, 
and Patrick Nicholson and Pawel Gawrychowski for pointing out mistakes in an earlier version.

\end{document}